\documentclass[a4paper,6pt,intlimits]{iopart}
\expandafter\let\csname equation*\endcsname\relax
\expandafter\let\csname endequation*\endcsname\relax
\usepackage{graphicx}
\usepackage[tbtags]{amsmath}           
\usepackage{amsfonts}          
\usepackage{epstopdf}
\usepackage{amsthm}

\usepackage{subcaption}

\usepackage{newtxtext,newtxmath}

\usepackage{tikz,pgfplots,bbm,bm,mathrsfs,dsfont}
\usetikzlibrary{shapes,snakes,arrows,pgfplots.groupplots,fadings,calc,decorations.markings,positioning,chains,fit,shapes,calc}

\tikzset{
	>=triangle 45,
	box/.style={draw = black, rectangle, rounded corners, inner sep=4pt,fill=white,text=black, text centered},
	arro/.style={line width=1pt,gray,->,shorten <=2pt,shorten >=1pt}
}

\usepackage{color}
\usepackage{quoting}

\usepackage[english]{babel}
\usepackage{hyperref}
\hypersetup{pdfauthor={Caracciolo DiGioacchino Malatesta},pdftitle={Solution for a Bipartite Euclidean 2-matching in one dimension},%
	colorlinks, linktocpage=true, pdfstartpage=3, pdfstartview=FitV,%
	urlcolor=orange, linkcolor=blue, citecolor=red, 
}

\newcommand{\be}{\begin{equation}}
\newcommand{\ee}{\end{equation}}
\def\reff#1{(\protect\ref{#1})}

\newtheorem{lemma}{Lemma}

\newtheorem{pros}{Proposition}[section]

\usepackage{dsfont}
\def\reff#1{(\protect\ref{#1})}

\newcommand{\plas}{\mathfrak{p}}
\newcommand{\Pad}{\mathop{\rm Pad}}

\begin{document}

\title{Plastic number and possible optimal solutions for an Euclidean 2-matching in one dimension}
\author{Sergio Caracciolo}\ead{sergio.caracciolo@mi.infn.it}
\address{Dipartimento di Fisica, University of Milan and INFN, via Celoria 16, 20133 Milan, Italy}
\author{Andrea Di Gioacchino}\ead{andrea.digioacchino@unimi.it}
\address{Dipartimento di Fisica, University of Milan and INFN, via Celoria 16, 20133 Milan, Italy}
\author{Enrico M. Malatesta}\ead{enrico.m.malatesta@gmail.com}
\address{Dipartimento di Fisica, University of Milan and INFN, via Celoria 16, 20133 Milan, Italy}

\date{\today}
\begin{abstract}
In this work we consider the problem of finding the minimum-weight loop cover of an undirected graph. This combinatorial optimization problem is called 2-matching and can be seen as a relaxation of the traveling salesman problem since one does not have the unique loop condition. We consider this problem both on the complete bipartite and complete graph embedded in a one dimensional interval, the weights being chosen as a convex function of the Euclidean distance between each couple of points. Randomness is introduced throwing independently and uniformly the points in space. We derive the average optimal cost in the limit of large number of points. We prove that the possible solutions are characterized by the presence of ``shoelace'' loops containing 2 or 3 points of each type in the complete bipartite case, and 3, 4 or 5 points in the complete one. This gives rise to an exponential number of possible solutions scaling as $\plas^N$, where $\plas$ is the plastic constant. This is at variance to what happens in the previously studied one-dimensional models such as the matching and the traveling salesman problem, where for every instance of the disorder there is only one possible solution.

\end{abstract}

\section{Introduction}
Combinatorial optimization problems are a large class of problems in which one has to find in a finite and discrete space of configurations the one that minimizes an object function, called ``cost'' or ``energy'' function. Their interest in the physics community came in particular from their random version, in which some parameters of the cost function itself are random variables. In this case one is interested in evaluating average properties of the solution. This is at variance with the point of view of computational complexity where one focuses on the worst case scenario. In general, the typical instance of a random combinatorial optimization problem can be very different from the worst case~\cite{mertens2002computational}. However one can generate really hard instances tuning certain parameters of the model and observe abrupt changes of the typical computational complexity. Archetypal examples are the random K-SAT problem~\cite{monasson1997statistical, monasson1999determining} and the famous traveling salesman problem (TSP)~\cite{gent1996tsp}. Away from these critical values of parameters typical instances are, instead, easy to solve. This sudden change of behavior can be seen as phase transitions in physical systems~\cite{martina2001} and, for this reason, can be studied with techniques developed in statistical mechanics (see~\cite{monasson1999determining} and references therein). The general way of describing a random combinatorial optimization problem is to consider the cost function as the energy of a fictitious physical system at a certain temperature~\cite{kirkpatrick1983optimization,Sourlas1986,fu1986application}. Finding the minimum of the cost function is perfectly equivalent to study the low temperature properties of this physical system. Proceeding in this way, it turned out that, specially in mean field cases, the general theory of spin glasses and disordered systems could help not only to calculate those quantities at the analytical level using techniques like replica and cavity method~\cite{mezard1987spin}, but also to shed light on the design of new algorithms to find their solution~\cite{mezard2009information}. Celebrated is the result for the asymptotic value of the average optimal cost in the random assignment problem obtained by M\'{e}zard and Parisi~\cite{Mezard1985} using the replica method. The same result was obtained later via the cavity method~\cite{mezard1986mean, Parisi2001}. 

In this paper we consider a particular random combinatorial optimization problem called 2-matching (or 2-factor) which consists, given an undirected graph, in finding a spanning subgraph that contains only disjoint cycles. In Statistical Mechanics models on loops have been considered~\cite{Baxter}. In particular in two dimension loop coverings have been studied also in connection to conformal field theories (CFT), Schramm-Loewner evolution (SLE) and integrable models~\cite{morin2017two,Jesper}.

The 2-matching problem can be seen as a relaxation of the TSP, in which one has the additional constraint that there must be a unique cycle. We mention that both these problems can indeed be studied using replicas and the cavity method in infinite dimensions: one finds that, for large number of points, their average optimal cost is the same. However on the complete graph one can obtain a closed expression for the average optimal cost only using cavity method~\cite{Krauth1989}, since with the replica method~\cite{mezard1986replica} one has some unresolved technical problems.

Here we study the 2-matching problem in one dimension, both on the complete graph bipartitioning two sets of $N$ points and on the complete graph of $N$ vertices, throwing the points independently and uniformly in the compact interval $[0,1]$. The weights on the edges are chosen as a convex function of the Euclidean distance between adjacent vertices. Despite the fact that it is a one-dimensional problem, it is not a trivial one. In the following we show that, while almost for every instance of the disorder there is only one solution, by looking at the whole ensemble of instances there appears an exponential number of possible solutions scaling as $\plas^N$, where $\plas$ is the plastic constant. This is at variance with what happens for other random combinatorial optimization problems, like the matching problem and the TSP that were studied so far~\cite{Caracciolo:159,Caracciolo:160,Caracciolo:169,Caracciolo:171}. 
In both cases we know that, for every realization of the disorder, the configuration that solves the problem is unique.

The rest of the paper is organized as follows: in Sect.~\ref{sec:model} we give some definitions and we present our model in more detail. In Sect.~\ref{sec:2-matching} we write the cost of the 2-matching in terms of permutations and for every number of points we compare its cost with that of matching and TSP. We argue that, in the thermodynamic limit, its cost is twice the cost of the optimal matching. In Sect.~\ref{sec:solution} we characterize, for every number of points, the properties of the optimal solution. We compare it with the corresponding one of the TSP problem and we conclude that the number of possible solutions grows exponentially with $N$. In Sect.~\ref{sec:cost} we derive some upper bounds on the average optimal cost and in Sect.~\ref{sec:numerical} we compare them with numerical simulations, describing briefly the algorithm we have used to find numerically the solution. We study the finite-size corrections to the asymptotic average cost in the complete bipartite case, and the leading order in the complete case. Finally, in Sect.~\ref{sec:conclusions} we give our conclusions.

\section{The model}\label{sec:model}

Given a generic (undirected) graph $\mathcal{G} = (\mathcal{V}, \mathcal{E})$, a {\em factor} is a subgraph spanning on all the vertices, a $k$-{\em factor} is a factor $k$-regular, that is in which each vertex belongs exactly to $k$ edges. From now on we shall restrict to the case of simple graphs, i.e. undirected graphs in which self-loops, that is edges that connect a vertex to itself, and multiple edges, that is the possibility that two vertices are connected by more than one edge, are avoided. The adjacency matrix $A$ of a $k$-factor on a simple graph, which is symmetric by construction, has exactly $k$ entries $1$ in each row and therefore in each column, i.e. has to satisfy the constraints
\begin{equation}
\begin{aligned}
& \sum_{j=1}^{\left| \mathcal{V} \right|} A_{ij} = k\,, \qquad &  i \in \left[\left| \mathcal{V} \right|\right]\\
& A_{ij} \in \left\{ 0, 1 \right\} \,, \qquad & A_{ij} \leq G_{ij}
\end{aligned}
\label{ConstraintsGeneral}
\end{equation}
where $G$ is the adjacency matrix of the whole graph $\mathcal{G}$.

A 1-factor is a perfect matching, or a covering of the graph by disjoint dimers. A 2-factor is a perfect 2-matching, or a covering of the graph by disjoint loops. When a 2-factor is formed by only one loop, this is an {\em Hamiltonian cycle}.


Let us denote by $\mathcal{M}_2$ the set of 2-factors of the graph $\mathcal{G}$. Let us suppose now that a weight $w_e > 0$ is assigned to each edge $e \in \mathcal{E}$ of the graph $\mathcal{G}$. We can associate to each 2-factor $\nu\in \mathcal{M}_2$ a total cost
\be
E(\nu) :=  \sum_{e\in \nu} w_e \, .\label{E}
\ee
In the (weighted) 2-matching problem we search for the 2-factor $\nu^*\in \mathcal{M}_2$ such that the total cost in~\reff{E} is minimized,  that is
\be
E(\nu^*) = \min_{ \nu\in \mathcal{M}_2} E(\nu)\, . \label{nu^*}
\ee
If $\mathcal{H}$ is the set of Hamiltonian cycles for the graph $\mathcal{G}$, of course $\mathcal{H} \subset \mathcal{M}_2$ and therefore if $h^*$ is the optimal Hamiltonian cycle, we have
\begin{equation}
\label{Inequality}
E[h^*] \geq E[\nu^*] \,.
\end{equation}
One can assign the weights $w_e$ in different ways. For example, consider when the complete graph $\mathcal{K}_N$ is embedded in $[0,1]^d \subset \mathbb{R}^d$, that is at each $i\in [N] =\{1,2,\dots,N\}$ we associate a point $x_i\in [0,1]^d$, and for each $e=(i,j)$ with $i,j \in [N]$ we introduce a cost which is a function of their Euclidean distance
\be
w_e = |x_i-x_j|^p \,  \label{p}
\ee
with $p\in \mathbb{R}$. 
Analogously for the complete bipartite graph $\mathcal{K}_{N,N}$ we have two sets of points in $[0,1]^d$, that is, say, the reds $\{r_i\}_{i\in [N]}$ and the blues $\{b_i\}_{i\in [N]}$, and the edges connect red points with blue points with a cost
\be
w_e = |r_i-b_j|^p \, . \label{pb}
\ee
In the {\em random} 2-matching problem, the weights $w_e$'s are random variables.
In this case, the typical properties of the optimal solution are of interest, and in particular the average optimal cost
\be
\overline{E}  := \overline{E(\nu^*)} \,,
\ee
where we have denoted by a bar the average over all possible instances of the costs set. The simplest way to introduce randomness in the problem is to consider the weights $w_e$ independent and identically distributed random variables. For a discussion on this problem on an arbitrary graph $\mathcal{G}$, see~\cite{Zecchina} and references therein.


\begin{figure}[ht]
	\begin{subfigure}[t]{0.48\linewidth}
	\centering
	\includegraphics[width=1\columnwidth]{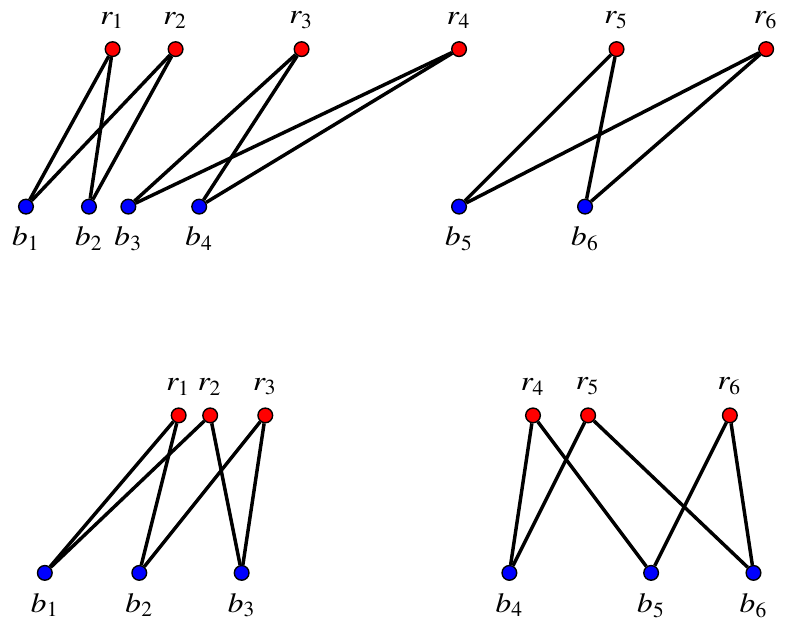}
	\caption{\footnotesize Two instances whose optimal solutions are the two possible $\nu^*$ for $N=6$ on the complete bipartite graph $\mathcal{K}_{N,N}$. For each instance the blue and red points are chosen in the unit interval and sorted in increasing order, then plotted on parallel lines to improve visualization.} \label{N=6}
\end{subfigure} \hfill
\begin{subfigure}[t]{0.48\linewidth}
	\centering
	\includegraphics[width=1\columnwidth]{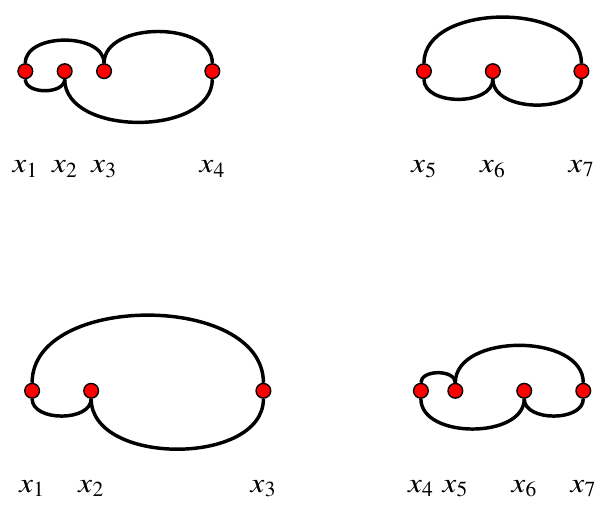}
	\caption{\footnotesize Two instances whose optimal solutions are the two possible  $\nu^*$ for $N=7$ on the complete graph $\mathcal{K}_N$. For each instance the points are chosen in the unit interval and sorted in increasing order. }  \label{N=7_mono}
\end{subfigure}
\caption{Optimal solutions for small $N$ cases.}
\end{figure}

\section{The Euclidean 2-matching problem}\label{sec:2-matching}

In the random Euclidean 2-matching problem the positions of the points are generated at random and, as a consequence, the weights will be correlated.
Let us start by making some considerations when the problem is defined on the complete bipartite graph $\mathcal{K}_{N,N}$, where each cycle must have an even length.

Let $\mathcal{S}_N$ be the symmetric group of order $N$ and consider two permutations $\sigma, \pi \in \mathcal{S}_N$. If for every $i\in [N]$ we have that $\sigma(i) \neq \pi(i)$, then the two permutations define the 2-factor $\nu(\sigma,\pi)$ with edges
\begin{align}
e_{2i-1} \; := \; & (r_i, b_{\sigma(i)})\\
e_{2i} \; := \; & (r_i, b_{\pi(i)})
\end{align}
for $i\in[N]$. And, viceversa, for any 2-factor $\nu$ there is a couple of permutations  $\sigma, \pi \in \mathcal{S}_N$,  such that for every $i\in [N]$ we have that $\sigma(i) \neq \pi(i)$.

t will have total cost
\be
E[\nu(\sigma, \pi)] = \sum_{i\in [N]} \left[ |r_i - b_{\sigma(i)}|^p +  |r_i - b_{\pi(i)}|^p \right] \, .
\ee
By construction, if we denote by $\mu[\sigma]$ the matching associated to the permutation $\sigma$ and by
\be
E[\mu(\sigma)] := \sum_{i\in [N]}  |r_i - b_{\sigma(i)}|^p
\ee
its cost, we soon have that
\be
E[\nu(\sigma, \pi)] = E[\mu(\sigma)] + E[\mu(\pi)]
\ee
and we recover that
\be
\label{bip::inequality}
E[\nu^*] \geq 2\, E[\mu^*]
\ee
the cost of the optimal 2-factor is necessarily greater or equal to twice the optimal 1-factor. Together with inequality~(\ref{Inequality}), which is valid for any graph, we obtain that
\be
\label{Inequalities} 
E[h^*] \geq E[\nu^*] \geq 2\, E[\mu^*] \, .
\ee
In~\cite{Caracciolo:171} we have seen that in the limit of infinitely large $N$, in one dimension and with $p>1$, the average cost of the optimal Hamiltonian cycle is equal to twice the average cost of the optimal matching (1-factor). We conclude that the average cost of the 2-matching must be the same. In the following we will denote with $\overline{E_{N,N}^{(p)}[\nu^*]}$ the average optimal cost of the 2-matching problem on the complete bipartite graph. Its scaling for large $N$ will be the same of the TSP and the matching problem, that is the limit
\begin{equation}
\lim\limits_{N\to \infty} \frac{\overline{E_{N,N}^{(p)}[\nu^*]}}{N^{1-p/2}} = E^{(p)}_B \,,
\end{equation}
is finite. An explicit evaluation in the case $p=2$ is presented in Sec.~\ref{sec:cost}.

On the complete graph $\mathcal{K}_N$ inequality (\ref{bip::inequality}) does not hold, since a general 2-matching configuration cannot always be written as a sum of two disjoint matchings, due to the presence of odd-length loops. Every 2-matching configuration on the complete graph can be determined by only one permutation $\pi$, satisfying $\pi(i) \ne i$ and $\pi(\pi(i)) \ne i$ for every $i\in [N]$. The cost can be written as
\begin{equation}
E[\nu(\pi)] = \sum_{i\in [N]}  |x_i - x_{\pi(i)}|^p \,.
\end{equation}
The two constraints on $\pi$ assure that the permutation does not contain fixed points and cycles of length 2. 
In the following we will denote with $\overline{E_{N}^{(p)}[\nu^*]}$ the average optimal cost of the 2-matching problem on the complete graph. Even though inequality~(\ref{bip::inequality}) does not hold, we expect that for large $N$, the average optimal cost scales in the same way as the TSP and the matching problem, i.e. as
\begin{equation}
\lim\limits_{N \to \infty} \frac{\overline{E_{N}^{(p)}[\nu^*]}}{N^{1-p}} = E^{(p)}_M \,.
\end{equation}
In Sect.~\ref{sec:numerical} we give numerical evidence for this scaling.

\begin{figure*}[t]
	\begin{subfigure}[t]{0.49\linewidth}
		\centering
		\includegraphics[width=0.9\columnwidth]{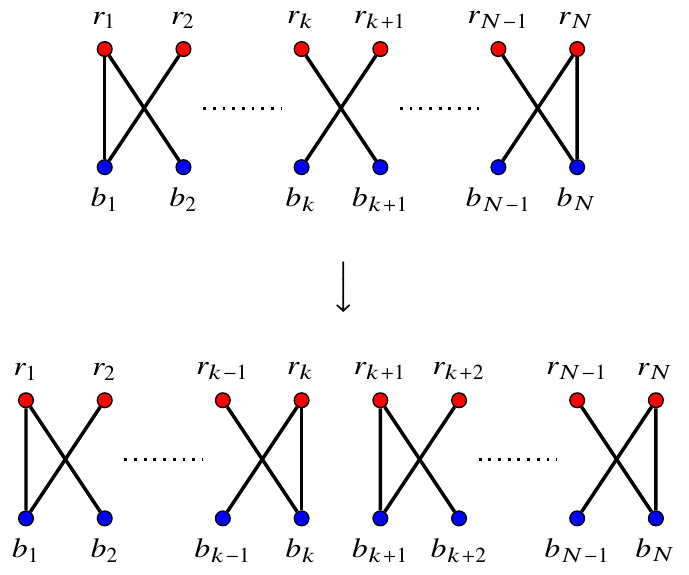}
		\caption{\footnotesize$\mathcal{K}_{N,N}$ case}
		\label{Fig::bip}
	\end{subfigure} \hfill
	\begin{subfigure}[t]{0.49\linewidth}
		\centering
		\includegraphics[width=1\columnwidth]{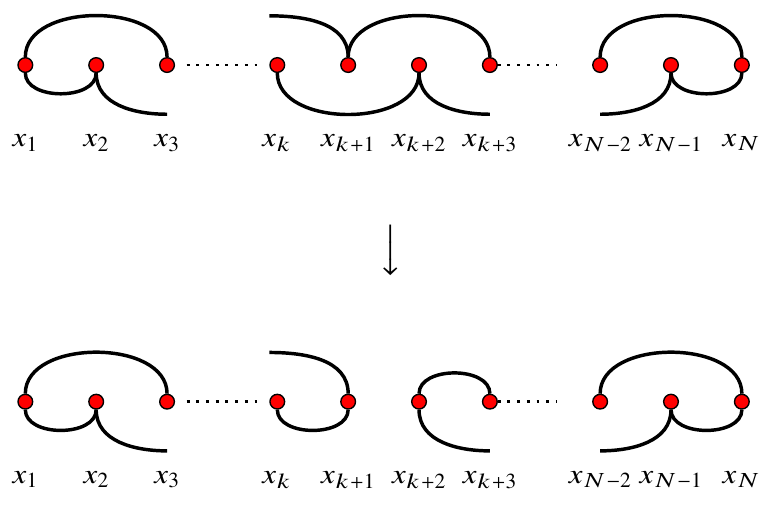}
		\caption{\footnotesize $\mathcal{K}_N$ case}
		\label{Fig::mono}
	\end{subfigure}
	\caption{Result of one cut of the shoelace in two smaller ones for both the complete bipartite and complete graph cases. The cost gained is exactly the difference between an unordered matching and an ordered one.}
\end{figure*}

\section{Properties of the solution for $d=1$}\label{sec:solution}

We restrict here to the particular case in which the parameter $p$ appearing in the definition of the cost~\reff{p} is such that $p>1$, that is the weight associated to an edge is a convex and increasing function of the Euclidean distance between its two vertices. 
In such a case we know exactly, for every number of points, the optimal solution of the matching problem both on the bipartite~\cite{Caracciolo:159, Caracciolo:160, Caracciolo:169} and the complete graph~\cite{Caracciolo:169} and of the TSP problem, again on both its bipartite~\cite{Caracciolo:171} and complete graph version~\cite{CDMV}. The knowledge of the optimal configuration of those problems permits to write down several properties of the solution of the 2-matching.  

\subsection{Bipartite Case}
The adjacency matrix of a bipartite graph with the same cardinality $N$ 
of red and blue points can always be written in block form
\begin{equation}
A = \left(
\begin{array}{cc}
\bold 0 & B \\
B^T & \bold 0
\end{array}
\right) \,,
\label{BipartiteAdj}
\end{equation}
where $B$ is a $N \times N$ matrix containing the only non-zero entries corresponding to edges connecting the two different types of points.
Of course $B$ uniquely identifies the adjacency matrix $A$.

Now, suppose that both blue and red points are labeled in increasing order, that is if $i>j$ with $i,j\in [N]$, then $r_i>r_j$ and $b_i>b_j$,  the permutation which minimizes the cost of the matching is necessarily the identity permutation $\mu^*(i) = i$ for $i\in [N]$, so that
\be
E[\mu^*] = \sum_{i\in [N]} |r_i - b_i|^p \, .
\ee
The optimal matching corresponds to
\be
B = \left(
\begin{array}{cccc}
 1 & \cdots & \cdots &  0 \\
 \vdots & \ddots &  & \vdots  \\
 \vdots &   & \ddots  & \vdots  \\ 
 0 & \cdots & \cdots &  1
\end{array}
\right) \, .
\ee
Since the total adjacency matrix is of the form (\ref{BipartiteAdj}), $B$ has to satify the following constraints
\begin{subequations}
	\begin{align}
	& \sum_{i=1}^{N} B_{ij} = 1\,, \qquad j \in \left[N\right]\\
	& \sum_{j=1}^N B_{ij} = 1\,, \qquad i \in \left[N\right]\\
	& B_{ij} \in \left\{ 0, 1 \right\} \,.
	\end{align}
	\label{Matching}
\end{subequations}
The first two constraints impose that only one edge must depart from each blue and each red vertex respectively.
For the TSP, the optimal Hamiltonian cycle $h^*$ is derived from the two permutations $\tilde{\sigma}$ and $\tilde{\pi}$ defined by
\be
\tilde{\sigma}(i) = 
\begin{cases}
2i-1 & i \leq  (N+1)/2 \\
2N -2i +2 & i > (N+1)/2 \label{sigmatilde}
\end{cases}
\ee
and
\be
\tilde{\pi}(i) = \tilde{\sigma}(N+1-i)  =
\begin{cases}
2i  & i < (N +1)/2 \\
2N - 2i +1 & i \geq  (N +1)/2 \label{pitilde}
\end{cases}
\end{equation}
for all $i\in [N]$. The corresponding Hamiltonian cycle, given by
$(r_1, b_2, r_3,\dots, b_3, r_2, b_1, r_1)$, that is called the {\em shoelace} configuration, has an optimal  cost
\begin{equation}
\begin{aligned}
E[h^*] = |r_1-b_1|^p + |r_N-b_N|^p +\sum_{i=1}^{N-1}\left[ |b_{i+1} - r_i|^p  + |r_{i+1} - b_i|^p \right] \label{EN}.
\end{aligned}
\end{equation}
The optimal Hamiltonian cycle corresponds to the adjacency matrix
\be
B = \left(
\begin{array}{cccccc}
 1 & 1 & 0 & 0 & \cdots &  0 \\
 1 & 0 & 1 & 0 & \cdots &  0 \\
 0 & 1 & 0 & 1 & \cdots &  0 \\
 \vdots &  \ddots & \ddots & \ddots & \ddots & \vdots \\
 0 & \cdots & 0 & 1 & 0 & 1 \\
 0 & \dots & 0 & 0 & 1 & 1
\end{array}
\right) \, .
\ee
Let us now look for the optimal solutions for the 2-matching.
The adjacency matrix of a valid 2-matching must satisfy constraints analogous to those of the matching problem, i.e.
\begin{subequations}
	\begin{align}
	& \sum_{i=1}^{N} B_{ij} = 2\,, \qquad j \in \left[N\right]\\
	& \sum_{j=1}^N B_{ij} = 2\,, \qquad i \in \left[N\right]\\
	& B_{ij} \in \left\{ 0, 1 \right\} \,.
	\end{align}
	\label{2-Matching}
\end{subequations}
The only difference with (\ref{Matching}) is that from every blue or red vertex must depart two edges.

For $N=2$ there is only one configuration. It can be defined by the adjacency matrix
\be
B_2= \left(
\begin{array}{cc}
 1 & 1 \\
 1 & 1 \\
\end{array}
\right)\, .
\ee
For $N=3$ the solution is the same as in the TSP
\be
B_3 = \left(
\begin{array}{ccc}
 1 & 1 & 0 \\
 1 & 0 & 1 \\
 0 & 1 & 1 \\
\end{array}
\right)\, .
\ee
For $N=4$ the solution has two simple cycles 
\be
B_2^{(2)}= 
\left(
\begin{array}{cccc}
 1 & 1 & 0 & 0 \\
 1 & 1 & 0 & 0 \\
 0 & 0 & 1 & 1 \\
 0 & 0 & 1 & 1 \\
\end{array}
\right)\, .
\ee
For $N=5$ there are two symmetric possible solutions
\be
B_{2,3} = \left(
\begin{array}{ccccc}
 1 & 1 & 0 & 0 & 0 \\
 1 & 1 & 0 & 0 & 0 \\
 0 & 0 & 1 & 1 & 0 \\
 0 & 0 & 1 & 0 & 1 \\
 0 & 0 & 0 & 1 & 1 \\
\end{array}
\right) \qquad
B_{3,2} = \left(
\begin{array}{ccccc}
 1 & 1 & 0 & 0 & 0 \\
 1 & 0 & 1 & 0 & 0 \\
 0 & 1 & 1 & 0 & 0 \\
 0 & 0 & 0 & 1 & 1 \\
 0 & 0 & 0 & 1 & 1 \\
\end{array}
\right) \, .
\ee
For $N=6$ there are two possible solutions (not related by symmetry)
\be
B_2^{(3)} = \left(
\begin{array}{cccccc}
 1 & 1 & 0 & 0 & 0 & 0 \\
 1 & 1 & 0 & 0 & 0 & 0 \\
 0 & 0 & 1 & 1 & 0 & 0 \\
 0 & 0 & 1 & 1 & 0 & 0 \\
 0 & 0 & 0 & 0 & 1 & 1 \\
 0 & 0 & 0 & 0 & 1 & 1 \\
\end{array}
\right)  \qquad
B_3^{(2)} = \left(
\begin{array}{cccccc}
 1 & 1 & 0 & 0 & 0 & 0 \\
 1 & 0 & 1 & 0 & 0 & 0 \\
 0 & 1 & 1 & 0 & 0 & 0 \\
 0 & 0 & 0 & 1 & 1 & 0 \\
 0 & 0 & 0 & 1 & 0 & 1 \\
 0 & 0 & 0 & 0 & 1 & 1 \\
\end{array}
\right) \, .
\ee

The possible solutions for $N=6$ and are represented schematically in Fig. \ref{N=6}. For $N=7$ there are three solutions and so on. 

\begin{lemma}
In any optimal 2-matching $\nu^*$ all the loops must be in the shoelace configuration.
\end{lemma}
\begin{proof}
In each loop there is the same number of red and blue points. Our general result for the TSP~\cite{Caracciolo:171} shows indeed that the shoelace loop is always optimal when restricted to one loop.
\end{proof}

\begin{lemma}
In any optimal 2-matching $\nu^*$ there are no loops with more than 3 red points.
\end{lemma}
\begin{proof}
As soon as the number of red points (and therefore blue points) in a loop is larger than 3, a more convenient 2-matching is obtained by considering a 2-matching with two loops. In fact, as can be seen in Fig.~\ref{Fig::bip}, the cost gain is exactly equal to the difference between an ordered and an unordered matching which we know is always negative for $p>1$ \cite{Caracciolo:159}.
\end{proof}

It follows that
\begin{pros}
\label{2and3}
In any optimal bipartite 2-matching $\nu^*$ there are only shoelaces loops with 2 or 3 red points.
\end{pros}

In different words to the optimal bipartite 2-matching solution $\nu^*$ is associated an adjacency matrix which is a block matrix built with the  sub-matrices $B_2$ and $B_3$. Two different 2-matchings in this class are not comparable, that is all of them can be optimal in particular instances.

\begin{pros}
\label{number}
At given number $N$ of both red and blue points there are at most $\Pad(N-2)$  optimal 2-matching $\nu^*$.
\end{pros}

$\Pad(N)$ is the $N$-th {\em Padovan} number, see the~\ref{A}, where it is also shown in~\reff{asym} that for large $N$
\be
\Pad(N) \sim \plas^N
\ee
with $\plas$ the {\em plastic} number~\reff{plastic} (see~\ref{B} for a discussion on this constant). 

Actually, for values of $N$ which we could explore numerically, we saw that all $\Pad(N-2)$ possible solutions appear as optimal solutions in the ensemble of instances.

\subsection{Complete Case}
Similar conclusions can be derived in the case of the complete graph $\mathcal{K}_N$ since, as we have said, both the analytical solution for the matching \cite{Caracciolo:169} and the TSP \cite{CDMV} are known. Let us order the points in increasing order, i.e. $x_i>x_j$ if $i>j$ with $i$, $j \in [N]$. In the matching problem on the complete graph the number of points must be even, and with $p>1$ the solution is very simple: if $j>i$ then the point $x_i$ will be matched to $x_j$ if and only if $i$ is odd and $j=i+1$ that is
\begin{equation}
E[\mu^*] = \sum_{i\in [N]} |x_{2i} - x_{2i-1}|^p \,.
\end{equation}
The corresponding adjacency matrix assumes the block diagonal form
\be
A = \left(
\begin{array}{ccccc}
	\bold a & \bold 0 & \cdots & \cdots & \bold 0 \\
	\bold 0 & \bold a & \cdots & \cdots & \bold 0 \\
	\vdots & \vdots & \ddots &  & \vdots \\
	\bold 0 & \bold 0 &   & \bold a & \bold 0  \\ 
	\bold 0 & \bold 0 & \cdots & \bold 0 &  \bold a
\end{array}
\right) \,,
\label{adjMatching}
\ee
where 
\be
\bold a = \left(
\begin{array}{cc}
0 & 1 \\
1 & 0 
\end{array}
\right) \,.
\ee
The adjacency matrix (\ref{adjMatching}) satisfies constraints (\ref{ConstraintsGeneral}), with $k=1$.
In the case of the TSP on the complete graph, where the number of points can also be odd, the optimal permutation is the same $\tilde \sigma$ defined in~(\ref{sigmatilde})
\begin{equation}
h^* = \left( x_{\tilde \sigma(1)}, x_{\tilde \sigma(2)}, \dots , x_{\tilde \sigma(N)} \right) \,.
\end{equation}
With a slightly abuse of language we will call ``shoelace'' also the optimal loop configuration for the TSP problem on complete graph. The adjacency matrix is
\be
A = \left(
\begin{array}{ccccccc}
	0 & 1 & 1 & 0 & \cdots & \cdots & 0 \\
	1 & 0 & 0 & 1 & \cdots & \cdots & 0 \\
	1 & 0 & 0 & 0 & 1 & \cdots & 0 \\
	\vdots &  \ddots & \ddots & \ddots & \ddots & \ddots & \vdots \\
	0 & \cdots & 1 & 0 & 0 & 0 & 1 \\
	0 & \cdots & 0 & 1 & 0 & 0 & 1\\
	0 & \dots & 0 & 0 & 1 & 1 & 0
\end{array}
\right) \, .
\ee

The possible solutions for the 2-matching on complete graph can be constructed by cutting in a similar way the corresponding TSP solution into smaller loops as can be seen pictorially in Fig.~\ref{Fig::mono}. Note that one cannot have a loop with two points. Analogously to the bipartite case we have analyzed before, each loop that form the 2-matching configuration must be a shoelace. However the length of allowed loops will be different, since one cannot cut, on a complete graph, a TSP of 4 and 5 points in two smaller sub-tours. It follows that
\begin{pros}
	\label{3and4and5}
	On the complete graph, in the optimal 2-matching $\nu^*$ there are only loops with 3, 4 or 5 points.
\end{pros}

In other terms the optimal configurations are composed by the adjacency matrices $A_1$, $A_2$, $A_3$ that are
\begin{subequations}
\begin{align}
& A_3 = \left(
\begin{array}{ccc}
	0 & 1 & 1 \\
	1 & 0 & 1 \\
	1 & 1 & 0 \\
\end{array}
\right)\,, \\
& A_4 = \left(
\begin{array}{cccc}
	0 & 1 & 1 & 0\\
	1 & 0 & 0 & 1\\
	1 & 0 & 0 & 1\\
	0 & 1 & 1 & 0
\end{array}
\right)\,,
\end{align}
\begin{align}
& A_5 = \left(
\begin{array}{ccccc}
0 & 1 & 1 & 0 & 0 \\
1 & 0 & 0 & 1 & 0 \\
1 & 0 & 0 & 0 & 1 \\
0 & 1 & 0 & 0 & 1 \\
0 & 0 & 1 & 1 & 0
\end{array}
\right)\,.
\end{align}
\end{subequations}

In Fig. \ref{N=7_mono} we represent the two solutions when $N=7$.
In~\ref{C} we prove that, similarly to the bipartite case, the number of  2-matching solutions is at most $g_N$ on the complete graph, which for large $N$ grows according to
\begin{equation}
g_N \sim \plas^N \,.
\end{equation}
Also in this case we verified numerically, for accessible $N$, that the set of possible solutions that we have identified is actually realized by some instance of the problem.
\begin{figure*}
	\vspace{1cm}
	\begin{subfigure}[t]{0.48\linewidth}
		\centering\includegraphics[scale=0.5]{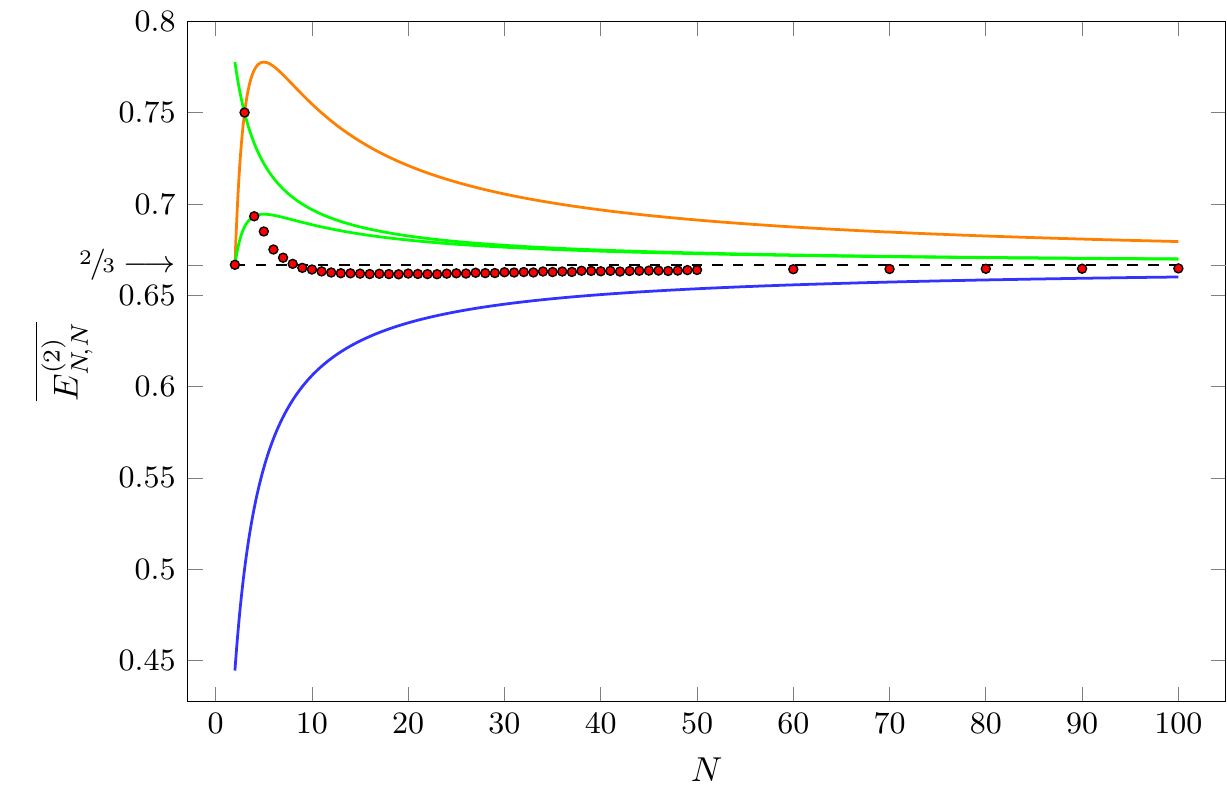} 
		\caption{\footnotesize $\mathcal{K}_{N,N}$ case. The orange line is the cost of the TSP given in~\reff{TSP}; the green lines are, from above, the cost of the optimal fixed 2-matching $\nu_{(2,2,\dots,2,3)}$ given in~\reff{odd} and $\nu_{(2,2,\dots,2)}$ given in~\reff{even}. The dashed black line is the asymptotic value $\frac{2}{3}$ and the blue continuous one is twice the cost of the optimal 1-matching $\frac{2}{3}\frac{N}{N+1}$. Red points are the results of a 2-matching numerical simulation, in which we have averaged over $10^7$ instances.} \label{plot}
	\end{subfigure} \hfill
	\begin{subfigure}[t]{0.48\linewidth}
		\centering\includegraphics[scale=0.5]{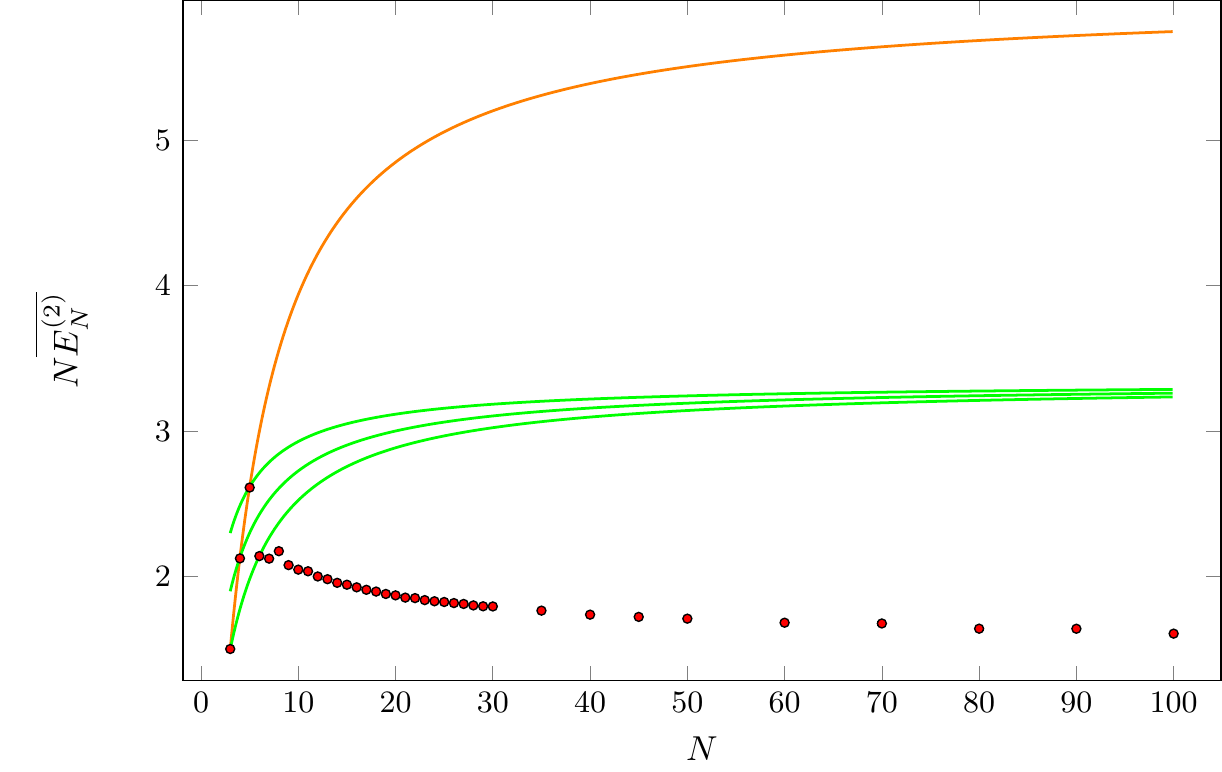} 
		\caption{\footnotesize $\mathcal{K}_N$ case.  Here the average cost is rescaled with $N$. The orange line is the cost of the TSP given in~\reff{TSP_Mono}. The green lines are from above the cost of the fixed 2-matching $\nu_{(3,3,\dots,3,5)}$ given in~\reff{multipleOfThree_plus2}, $\nu_{(3,3,\dots,4)}$ given in~\reff{multipleOfThree_plus1} and $\nu_{(3,3,\dots,3)}$ given in~\reff{multipleOfThree}. 
		Red points are the results of a numerical simulation for the 2-matching, in which we have averaged over $10^5$ instances for $N \le 30$, $10^4$ for $30 < N \le 50$ and $10^3$ for $N>50$.} \label{plotC}
	\end{subfigure}
	\caption{Average optimal costs for various $N$ and for $p=2$.}
\end{figure*}

\section{Bounds on the cost}\label{sec:cost}

Here we will derive the consequences of the results of the previous section, obtaining explicitly some upper bounds on the average optimal cost of the 2-matching problem. We will examine the complete bipartite case first, where we consider, for simplicity, the $p=2$ case~\cite{Caracciolo:171}. Indeed the calculation we perform below can be done also for general $p>1$, but it is much more involved. 
Then we will examine the complete graph case, where we have obtained a very simple expression of the average optimal cost for every $N$, and for every $p>1$~\cite{CDMV}.

\subsection{Bipartite Case}
Let us analyze the problem on the complete bipartite graph $\mathcal{K}_{N,N}$. In~\cite{Caracciolo:171} we derived for $p=2$ the exact result for all $N$ of the TSP when all the points are chosen with a flat distribution in the interval $[0,1]$
\be 
\overline{E_{N,N}^{(2)}[h^*]} = \frac{2}{3} \frac{N^2 + 4 N -3}{(N+1)^2} \label{TSP}
\ee
from which we soon obtain that
\be
\overline{E_{N,N}^{(2)}[\nu^*]} =  
\begin{cases}
\frac{2}{3}  & \hbox{for } N=2\\
\frac{3}{4}  & \hbox{for } N=3 \,,
\end{cases}
\ee
since in the cases $N=2$ and $N=3$ the solutions are the same as in the TSP. \, For $N=4$ we have still only one solution, which corresponds to two cycles on the first and the last 2 red points. Both cycles have the same cost and we easily get that
\be
\overline{E_{4,4}^{(2)}[\nu^*]} =  \frac{52}{75} \, .
\ee
This result can be obtained also in a different way. We first remark that 
\be
\overline{(r_k - b_k)^2} + \overline{(r_{k+1} - b_{k+1})^2} - \overline{(r_{k} - b_{k+1})^2} - \overline{(r_{k+1} - b_{k})^2} = - \frac{2}{(N+1)^2} \label{cut}
\ee
irrespectively from the choice of $1\leq k \leq N-1$.  This is exactly the cost gained by cutting a longer cycle into two smaller ones at position $k$, see Fig.~\ref{Fig::bip}. Therefore the cost for the optimal 2-matching for $N=4$ is the cost for the optimal Hamiltonian cycle, which from~\reff{TSP} is $\frac{58}{75}$, decreased because of a cut, that is by $-\frac{2}{25}$.

For $N=5$ there are two possible optimal solutions that we will denote by $\nu_{(2,3)}$ and $\nu_{(3,2)}$. For both of them
\be
\overline{E_{5,5}^{(2)}[\nu_{(2,3)}]} =  \overline{E_{5,5}^{(2)}[\nu_{(3,2)}]} = \frac{13}{18}
\ee
and therefore 
\be
\overline{E_{5,5}^{(2)}[\nu^*]} = \overline{\min\left\{E_{5,5}^{(2)}[\nu_{(2,3)}], E_{5,5}^{(2)}[\nu_{(3,2)}]\right\}} \leq \frac{13}{18} \, .
\ee
For $N=6$ there are still two possible optimal solutions, that is $\nu_{(3,3)}$ and $\nu_{(2,2,2)}$,
but this time they have not the same average cost, indeed
\begin{align}
\overline{E_{6,6}^{(2)}[\nu_{(3,3)}]}  = & \, \frac{36}{49} = \frac{38}{49} - \frac{2}{49}\\
\overline{E_{6,6}^{(2)}[\nu_{(2,2,2)}]}  = & \, \frac{34}{49} = \frac{38}{49} - \frac{4}{49}
\end{align}
that we have written as the TSP value from~\reff{TSP} decreased, respectively, by one and two cuts~\reff{cut} for $N=6$.

Now it is clear that when $N$ is even the 2-matching with lowest average energy is $\nu_{(2,2,\dots,2)}$ and that
\be
\overline{E_{N,N}^{(2)}[\nu_{(2,2,\dots,2)}]} = \frac{2}{3} \frac{N^2 + 4 N -3}{(N+1)^2} - \frac{N-2}{(N+1)^2} = \frac{1}{3} \frac{N (2 N + 5)}{(N+1)^2}\,, \label{even}
\ee
which is an upper bound for the optimal average cost since, even though this configuration has the minimum average cost, for every fixed instance of disorder there can be another one which is optimal. For $N$ odd  one of the 2-matchings with lowest average energy is $\nu_{(2,2,\dots,2,3)}$ and
\be
\overline{E_{N,N}^{(2)}[\nu_{(2,2,\dots,2,3)}]} = \frac{2}{3} \frac{N^2 + 4 N -3}{(N+1)^2} - \frac{N-3}{(N+1)^2} = \frac{1}{3} \frac{2 N^2 + 5 N + 3}{(N+1)^2} \,, \label{odd}
\ee
a result which shows that essentially the upper bound for the optimal average cost for even and odd large $N$ is the same.

\subsection{Complete Case}

Let us now turn to the problem  on the complete graph. 
In~\cite{CDMV} it is shown that, for every $N$ and every $p>1$, the average optimal cost of the TSP has the expression
\begin{equation}
\overline{E_N^{(p)}[h^*]} = \left[ (N-2) (p+1) +2 \right] \, \frac{\Gamma(N+1)\, \Gamma(p+1)}{\Gamma(N+p+1)} \,. 
\label{TSP_Mono}
\end{equation}
An analogous expression is present in the case of the matching problem \cite{Caracciolo:169}, where the number of points $N$ is even
\begin{equation}
\overline{E_N^{(p)}[\mu^*]} = \frac{N \, \Gamma(N+1) \, \Gamma(p+1)}{2 \, \Gamma(N+p+1)} \,.
\label{matching_Mono}
\end{equation}

Let us now turn to the evaluation of the cost gain when we cut the cycle in two shoelaces sub-cycles. For $p>1$ the cost gain doing one cut can be written as
\begin{equation}
\begin{aligned}
\overline{\left(x_{k+1}-x_k\right)^p} + \overline{\left(x_{k+3}-x_{k+2}\right)^p} - \overline{\left(x_{k+3}-x_{k+1}\right)^p} - \overline{ \left(x_{k+2}-x_{k}\right)^p} = - \frac{2 \,p \, \Gamma(N+1) \, \Gamma(p+1)}{\Gamma(N+p+1)} \,.
\end{aligned}
\end{equation}
For example for $N=6$ (in which the solution is unique since 6 can be written as a sum of 3, 4 and 5 in an unique way as 3+3) and $p=2$ we have
\begin{equation}
\overline{E_6^{(2)}} = \frac{1}{2} - \frac{1}{7} = \frac{5}{14} \,.
\end{equation}
If $N$ is multiple of 3, the lowest 2-matching is, on average, the one with the largest number of cuts i.e. $\nu_{(3,3,\dots,3)}$. The number of cuts is $(N-3)/3$ so that the average cost of this configuration is
\begin{equation}
\begin{aligned}
\overline{E_N^{(p)}[\nu_{(3,3,\dots,3)}]} & 
= N \left( \frac{p}{3} + 1\right) \frac{\Gamma(N+1)\, \Gamma(p+1)}{\Gamma(N+p+1)} \,.
\label{multipleOfThree}
\end{aligned}
\end{equation}
Instead if $N$ can be written as a multiple of 3 plus 1, the minimum average energy configuration is $\nu_{(3,3,\dots ,3,4)}$, which has $(N-4)/3$ cuts and
\begin{equation}
\begin{aligned}
\overline{E_N^{(p)}[\nu_{(3,3,\dots,4)}]} & 
= \left[ N \left( \frac{p}{3} + 1\right) + \frac{2}{3}p \right]\frac{ \Gamma(N+1) \, \Gamma(p+1)}{\Gamma(N+p+1)} \,.
\label{multipleOfThree_plus1}
\end{aligned}
\end{equation}
The last possibility is when $N$ is a multiple of 3 plus 2, so the minimum average energy configuration is $\nu_{(3,3,\dots ,3,5)}$, with $(N-4)/3$ cuts and
\begin{equation}
\begin{aligned}
\overline{E_N^{(p)}[\nu_{(3,3,\dots,5)}]} & 
= \left[ N \left( \frac{p}{3} + 1\right) + \frac{4}{3}p \right]\frac{ \Gamma(N+1) \, \Gamma(p+1)}{\Gamma(N+p+1)} \,.
\label{multipleOfThree_plus2}
\end{aligned}
\end{equation}
In the limit of large $N$ all those three upper bounds behave in the same way. For example
\begin{equation}
\lim\limits_{N \to \infty} \overline{E_N^{(p)}[\nu_{(3,3,\dots,3)}]} = N^{1-p} \left( 1 + \frac{p}{3} \right) \Gamma(p+1) \,.
\end{equation} 
Note that the scaling of those upper bounds for large $N$ is the same of those of matching and TSP.


\begin{figure*}
	\begin{subfigure}[t]{0.48\linewidth}
		\centering\includegraphics[width=1\columnwidth]{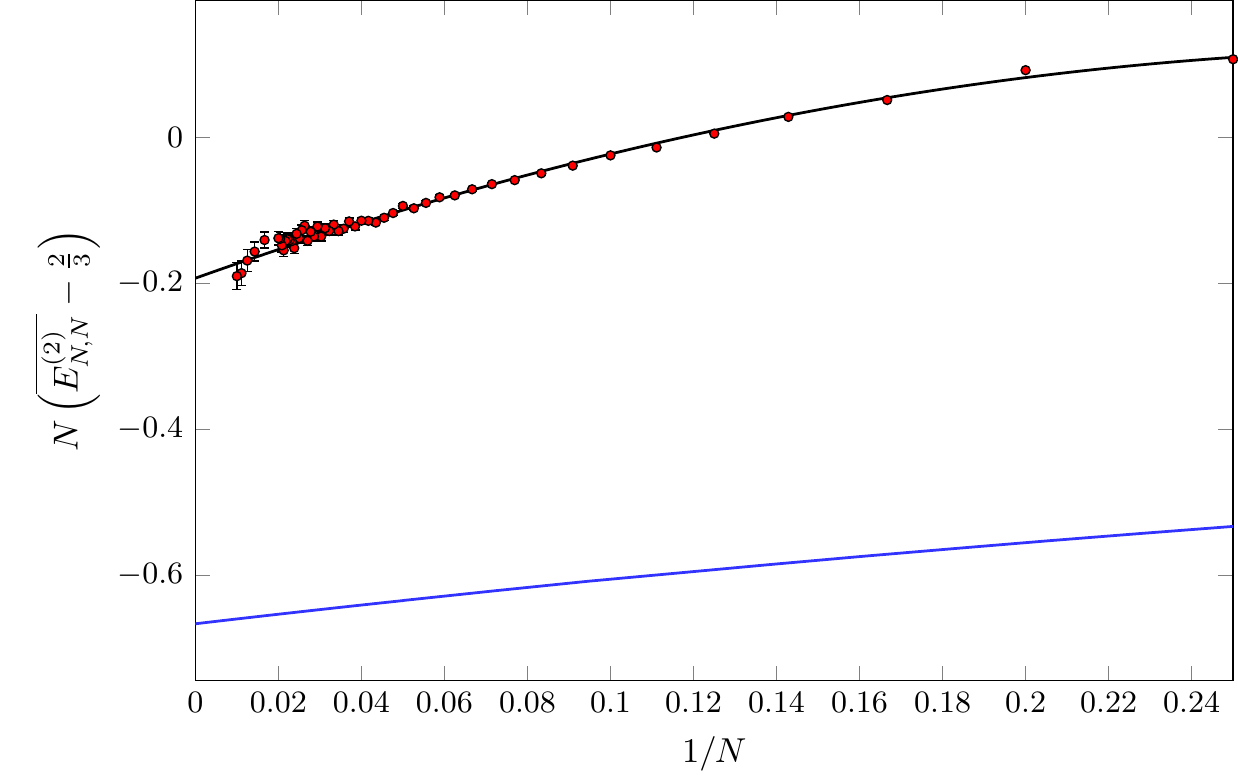} 
		\caption{Numerical results of $N \left(\overline{ E_{N,N}^{(2)} }-\frac{2}{3}\right)$ (red points) in the complete bipartite case; the black line is the fitting function (\ref{fit}). The blue line is two times the matching. The values of $a_1$, $a_2$ and $a_3$ are reported in Table~\ref{tab::fit}.} 
		\label{plot_FSC}
	\end{subfigure}\hfill
	\begin{subfigure}[t]{0.48\linewidth}
		\centering\includegraphics[width=1\columnwidth]{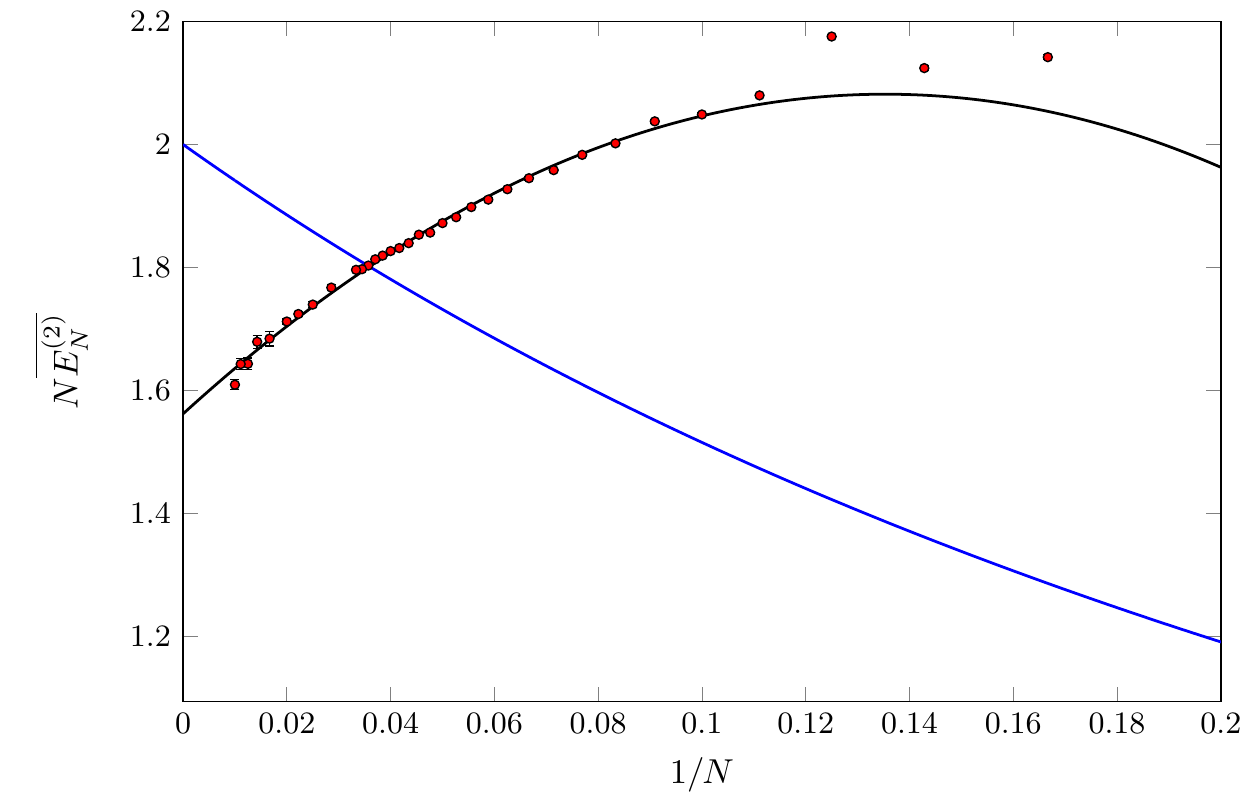} 
		\caption{Numerical results of $N \overline{ E_{N}^{(2)} }$ (red points) in the complete graph case; the black line is the fitting function (\ref{fit_mono}). The blue line is two times the value of the optimal matching as given in equation (\ref{matching_Mono}). The values of $b_0$, $b_1$ and $b_2$ are reported in Table~\ref{tab::fit_mono}.
		} \label{plot_FSC_mono}
	\end{subfigure}
	\caption{}
\end{figure*}

\section{Numerical Results}\label{sec:numerical}
In this section we present our numerical simulations describing briefly the algorithm we have used to find the solution for every instance of the problem. The 2-matching problem has an integer programming formulation. Given a generic simple graph $\mathcal{G} = (\mathcal{V}, \mathcal{E})$, the solution can be uniquely identified by a $\left| \mathcal{V} \right| \times \left| \mathcal{V} \right|$ matrix of occupation numbers $A_{ij}$ which can assume values 0 or 1. In particular $A_{ij}$ assumes value 0 if node $i$ is not connected to node $j$ in the 2-matching solution and 1 otherwise. The problem can be stated as the minimization of the energy function
\begin{equation}
E(A) = \frac{1}{2} \sum_{i=1}^{\left| \mathcal{V} \right|} \sum_{j=1}^{\left|\mathcal{V} \right|} A_{ij} \, w_{ij} \,,
\end{equation}
subject to the constrain (\ref{ConstraintsGeneral}).

We have performed some numerical simulations using a C++ code and the open source GLPK package, a library that solves general large scale linear programming problems. In Fig.~\ref{plot} and~\ref{plotC} we plot the results of some numerical simulations for $p=2$ respectively for the complete bipartite and complete graph case and we compare them with some exact results. In the complete graph case we plot $N \overline{E_N^{(2)}}$ revealing that the scaling of the cost is the same of the TSP and the matching problem. However the two situations are completely different, since in the complete case the bound estimate only gets worse when $N$ increases.

In order to understand the analytic form of the finite-size correction, we have also performed a parametric fit of the quantity $N \left(\overline{ E_{N,N}^{(2)} }-\frac{2}{3}\right)$ using a fitting function of the type
\begin{equation}
f_{B}(N) = a_1 + \frac{a_2}{N} + \frac{a_3}{N^2} \,.
\label{fit}
\end{equation}
In Fig. \ref{plot_FSC} we plot the numerical data and $f_B(N)$. The estimate of the parameters is reported in Table \ref{tab::fit}.

In the complete graph case, we have performed a fit of the rescaled cost $N \overline{ E_N^{(2)} }$ in order to evaluate numerically the asymptotic value of the cost. The fitting function was chosen to be
\begin{equation}
f_{M}(N) = b_0 + \frac{b_1}{N} + \frac{b_2}{N^2} \,.
\label{fit_mono}
\end{equation}
In Fig. \ref{plot_FSC_mono} we report the plot of the numerical data together with $f_M(N)$. Remember that in the complete graph case, the cost of the 2-matching cannot be bounded from below by two times the cost of the optimal matching as happens on the complete bipartite graph. For this reason in Fig. \ref{plot_FSC_mono} we have added the plot of the theoretical value of the optimal matching (given in equation  (\ref{matching_Mono})) multiplied by two.
The numerical values of the parameters are reported in Table \ref{tab::fit_mono}. Note also how in the complete graph case, the first finite-size correction $b_1$ is not only positive but its magnitude is much greater than $a_1$, its bipartite counterpart.


\begin{table}
	\begin{center}
			\begin{tabular}[b]{|c|c|c|}
				\hline
				$a_1$ & $a_2$ & $a_3$\\
				\hline
				$-0.193 \pm 0.003$ & $2.03 \pm 0.05$ & $-3.3 \pm 0.2$ \\
				\hline
			\end{tabular}
	\end{center}
	\caption{Numerical estimates of the parameters $a_1$, $a_2$ and $a_3$ defined in (\ref{fit}).}
	\label{tab::fit}
\end{table}

\begin{table}
	\begin{center}
			\begin{tabular}[b]{|c|c|c|}
				\hline
				$b_0$ & $b_1$ & $b_2$\\
				\hline
				$1.562 \pm 0.005$ & $ 7.7 \pm 0.2$ & $-28 \pm 2$ \\
				\hline
			\end{tabular}
	\end{center}
	\caption{Numerical estimates of the parameters $b_0$, $b_1$ and $b_2$ defined in (\ref{fit_mono}).}
	\label{tab::fit_mono}
\end{table}

\section{Conclusions}\label{sec:conclusions}
In this work we have examined the random Euclidean 2-matching problem in one dimension. We have considered the case in which the model is defined both on the complete bipartite graph and the complete graph, with a weight function which is a power $p$ of the Euclidean distance between the points. On the complete bipartite graph we have proved that in the convex case, i.e. $p>1$ and in the limit of large number of points $N$, the cost is equal to twice the cost of the optimal matching. 
Indeed, for every instance of the disorder, the cost of the 2-matching can be bounded from above by the TSP and from below by two times the cost of the corresponding optimal matching. 
An analogous bound from below lacks in the complete graph case. 
We have characterized the solution for every value of $N$ as a covering of shoelace loops with only 2 or 3 points of one color, in the complete bipartite case and 3, 4 and 5 points in the complete graph case. In the complete bipartite case this gives rise to $\Pad(N-2)$ possible optimal 2-matching. Therefore, in the large $N$ limit one has an exponential number $\plas^N$ of possible solutions, where $\plas$ is the plastic constant. A similar result holds on the complete graph. This is at variance to what happens in other one-dimensional problems when the weight function is convex. For example, in the matching and in the TSP problem on the complete bipartite graph for every instance of the disorder one has only one possible solution: in the matching case one has to match the $k$-th red point with the $k$-th blue one whereas in the TSP case the optimal Hamiltonian cycle is the shoelace configuration. 
Next we have derived some upper bounds on the average optimal cost by cutting the single shoelace loop configuration of the TSP problem in the maximum number of ways. Finally, we have performed some numerical simulations and we have compared them with analytical results and with theoretical curves of the TSP and matching. In the bipartite case 
we have studied numerically the form of the finite-size corrections. In the complete graph case, where we cannot bound the cost from below by the optimal matching, we have studied the large $N$ behavior of the average optimal cost. Our analysis essentially shows that also in the complete graph case the scaling of the cost is the same of the TSP and the matching problem.

In general the study of one-dimensional problems can help to shed light on their higher dimensional counterparts, where one does not really know how to properly treat Euclidean correlations. Recent progress include the study of the bipartite matching problem in $d>1$, where, by means of a scaling ansatz one can deduce not only the correct scaling of the cost but also the value of the average optimal cost and correlation functions in $d=2$~\cite{Caracciolo:158,Caracciolo:163,Caracciolo:162} and some predictions for the finite-size corrections in $d>2$~\cite{Caracciolo:158}. In addition, some of these results were also proven rigorously~\cite{Ambrosio2016}, thanks to the deep connection with optimal transport theory.
Here we have shown that the 2-matching problem, defined on both the complete bipartite and complete graph, is not a trivial model even in one dimension. 
An important question to investigate is if the relevant results we have found here, that connect tightly together matching, TSP and 2-matching problems in one dimension, continue to hold in higher dimension.  In~\cite{Caracciolo:174} we investigate the inequality~\reff{Inequalities}, which holds in any dimension $d$, but is, once more, saturated in $d=2$ on the complete bipartite graph.



\section*{Acknowledgments}
The authors thank Luca Guido Molinari for fruitful discussions. E.M.M wants to thank Giorgio Parisi for the many suggestions regarding the simulations performed. 

\appendix

\section{The Padovan numbers} \label{A}
According to Proposition \ref{2and3}, in the optimal 2-matching configuration of the complete bipartite graph there are only loops of length 2 and 3. Here we will count the number of possible optimal solutions for each value of $N$.
Let $f_N$ be the number of ways in which the integer $N$ can be written as a sum in which the addenda are only 2 and 3. For example, $f_4 = 1$ because $N=4$ can be written only as $2+2$, but $f_5=2$ because $N=5$ can be written as $2+3$ and $3+2$. We simply get the recursion relation
\be
f_N = f_{N-2} + f_{N-3} \label{rec}
\ee
with the initial conditions $f_2 = f_3 = f_4 = 1$. The $N$-th {\em Padovan number} $\Pad(N)$ is defined as $f_{N+2}$. Therefore it satisfies the same recursion relation~\reff{rec} but with the initial conditions $\Pad(0) = \Pad(1) =\Pad(2) =1$.


A generic solution of~\reff{rec} can be written in terms of the roots of the equation
\be
x^3 = x+1 \,.
\ee
There is one real root 
\be 
\plas = \frac{(9 + \sqrt{69})^\frac{1}{3} + (9 - \sqrt{69})^\frac{1}{3}}{ 2^\frac{1}{3} 3^\frac{2}{3}} \approx 1.324717957244746\dots 
\label{plastic}
\ee
known as the {\em plastic} constant and two complex conjugates roots
\be
\begin{aligned}
z_\pm & = \frac{ ( -1 \pm  i\, \sqrt{3}) (9 + \sqrt{69})^\frac{1}{3} +  ( -1 \mp  i\, \sqrt{3})(9 - \sqrt{69})^\frac{1}{3}}{ 2^\frac{4}{3} 3^\frac{2}{3}} \\
& \approx -0.662359\,\text{\dots} \pm i \, 0.56228\dots 
\end{aligned}
\ee
of modulus  less than unity. Therefore
\be
\Pad(N) = a\, \plas^N + b \, z_+^N + b^* \, z_-^N
\ee
and by imposing the initial conditions we get
\be
\Pad(N) = \frac{(z_+-1)(z_--1)}{(\plas-z_+)(\plas-z_-)}\, \plas^N + \frac{(\plas-1)(z_--1)}{(z_+-\plas)(z_+-z_-)}\, z_+^N +\frac{(\plas-1)(z_+-1)}{(z_--\plas)(z_--z_+)}\, z_-^N\,.
\ee
For large $N$ we get
\be
\Pad(N)  \sim \lambda \, \plas^N \label{asym}
\ee
with $\lambda\approx 0.722124\dots$ the real solution of the cubic equation
\be
23 \,t^3 - 23\, t^2 + 6\, t -1 = 0 \, .
\ee
In Fig.~\ref{fig1} we plot the Padovan sequence for a range of values of $N$ and its asymptotic expression.
\begin{figure}[h!]
\centering
\includegraphics[width=0.7\columnwidth]{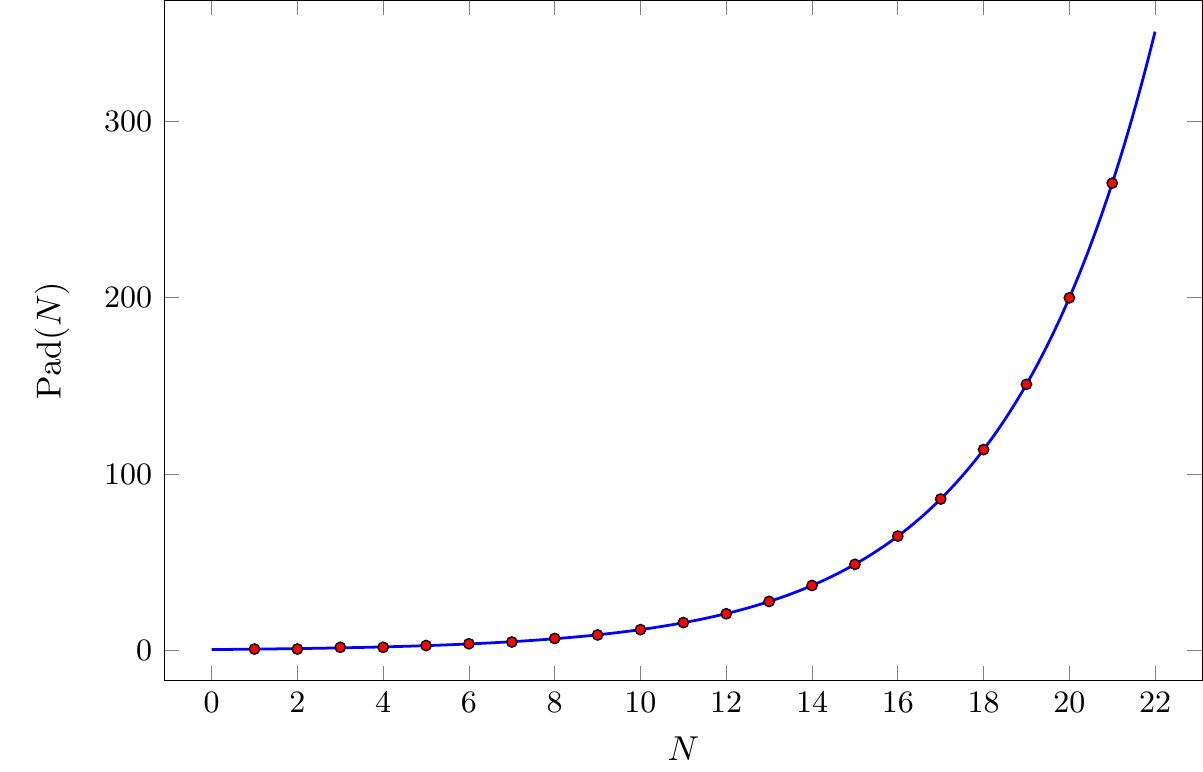} 
\caption{Padovan numbers and their asymptotic expansion.} \label{fig1}
\end{figure}

There is a relation between the Padovan numbers and the Binomial coefficients.
If we consider $k$ addenda equal to 3 and $s$ addenda equal to 2, there are $\binom{k+s}{k}=\binom{k+s}{s}$ possible different orderings. If we fix $N= 3 \, k + 2\, s$ we easily get that
\be
\Pad(N-2) = \sum_{k\ge 0} \sum_{s\ge 0} \delta_{N, 3 \, k + 2\, s } \, \binom{k+s}{k} = \sum_{m\ge 0} \sum_{k\ge 0} \delta_{N,  k + 2\, m } \, \binom{m}{k} \,.
\ee

\section{The recursion on the complete graph} \label{C}
A recursion relation analogous to eq. (\ref{rec}) can be derived for the number of possible solution of the 2-matching problem on the complete graph $\mathcal{K}_N$. Let $g_N$ be the number of ways in which the integer $N$ can be expressed as a sum of 3, 4 and 5. Then $g_N$ satisfies the recursion relation given by
\begin{equation}
g_N = g_{N-3} + g_{N-4} + g_{N-5} \,,
\end{equation}
with the initial conditions $g_3=g_4=g_5=g_6=1$ and $g_7 = 2$. The solution of this recursion relation can be written in function of the roots of the 5-th order polynomial
\begin{equation}
x^5-x^2-x-1=0 \,.
\end{equation}
This polynomial can be written as $(x^2+1)(x^3-x-1)=0$. Therefore the roots will be the same of the complete bipartite case ($\plas$, and $z_{\pm}$) and in addition 
\begin{equation}
y_{\pm} = \pm i \,.
\end{equation}
$g_N$ can be written as
\begin{equation}
g_N = \alpha_1 \plas^N + \alpha_2 z_+^N + \alpha_3 z_-^N + \alpha_4 y_+^N + \alpha_5 y_-^N \,,
\end{equation}
where the constants $\alpha_1$, $\alpha_2$, $\alpha_3$, $\alpha_4$, and $\alpha_5$ are fixed by the initial conditions $g_3=g_4=g_5=g_6=1$ and $g_7 = 2$. When $N$ is large the dominant contribution comes from the plastic constant
\begin{equation}
g_N \simeq \alpha_1 \plas^N \,.
\end{equation}
with $\alpha_1 \approx 0.262126..$.

\section{The plastic constant} \label{B}
In 1928,  shortly after abandoning his architectural studies and becoming a novice monk of the Benedictine Order,
Hans~van~der~Laan discovered a new, unique system of architectural proportions. Its construction is completely based on a single irrational value which he called the plastic number (also known as the plastic constant)~\cite{Maroni}.
This number was originally studied in 1924 by a French engineer, G. Cordonnier, when he was just 17 years old, calling it "radiant number". However, Hans van der Laan was the first who explained how it relates to the human perception of differences in size between three-dimensional objects and demonstrated his discovery in (architectural) design. His main premise was that the plastic number ratio is truly aesthetic in the original Greek sense, i.e. that its concern is not beauty but clarity of perception ~\cite{Padovan}. 
The word plastic was not intended, therefore, to refer to a specific substance, but rather in its adjectival sense, meaning something that can be given a three-dimensional shape~\cite{Padovan}. 
The golden ratio or divine proportion
\be
\phi = \frac{1+\sqrt{5}}{2} \approx 1.6180339887 \,,
\ee
which is a solution of the equation
\be
x^2 = x+1 \label{qe} \,,
\ee
has been studied by Euclid, for example for its appearance in the regular pentagon, and has been used to analyze the most aestetich proportions in the arts.
For example, the golden rectangle, of size $(a+b)\times a$ which may be cut into a square of size $a \times a$ and a smaller rectangle of size $b \times a$ with the same aspect ratio
\be
\frac{a+b}{a} = \frac{a}{b} = \phi \, .
\ee
This amounts to the subdivision of the interval $AB$ of length $a+b$ into $AC$ of length $a$ and $BC$ of length $b$. By fixing $a+b=1$ we get
\be
\frac{1}{a} = \frac{a}{1-a} = \phi \,,
\ee
which implies that $\phi$ is the solution of~\reff{qe}. The segments $AC$ and $BC$, of length, respectively $\frac{1}{\phi^2}(\phi, 1)$ are sides of a golden rectangle.

But the golden ratio fails to generate harmonious relations within and between three-dimensional objects. Van~der~Laan therefore elevates definition of the golden rectangle in terms of space dimension.
Van~der~Laan breaks segment $AB$  in a similar manner, but in three parts. If C and D are points of subdivision, plastic number $\plas$  is defined with
\be
\frac{ AB }{ AD } = \frac{AD}{BC} = \frac{BC}{ AC} = \frac {AC}{ CD} = \frac{CD}{BD} = \plas
\ee 
and by fixing $AB=1$, from $AC = 1 -BC$, $BD = 1 - AD$ we get
\be
\plas^3 = \plas +1 \,.
\ee
The segments $AC$, $CD$ and $BD$, of length, respectively, $\frac{1}{(\plas+1)\plas^2} (\plas^2, \plas, 1)$ can be interpreted as sides of a cuboid analogous to the golden rectangle.

\section*{References}
\bibliographystyle{iopart-num}
\bibliography{AssignmentANDTsp}

\end{document}